\documentclass[a4paper,11pt]{amsart}
\usepackage[french,english]{babel} 
\usepackage{a4wide,etex}
\usepackage[utf8]{inputenc}
\usepackage[T1]{fontenc}

\usepackage{bm,amssymb,comment,wrapfig,graphicx}

\newcommand{\bbC}{\mathbb{C}}

\newcommand{\be}{\begin{equation}}
\newcommand{\ee}{\end{equation}}
\newcommand{\til}[1]{\widetilde{#1}}
\newcommand{\N}{\mathbb{N}}

\newcommand{\R}{\mathbb{R}}
\newcommand{\mc}[1]{\mathcal{#1}}
\newcommand{\p}{\partial}
\newcommand{\e}{\varepsilon}
\newcommand{\dl}{\delta}

\newcommand{\re}{{\rm Re}\hskip 1pt }
\newcommand{\im}{{\rm Im}\hskip 1pt }
\newcommand{\ord}{{\mathcal O}}
\newcommand{\ope}[1]{\operatorname{#1}}
\newcommand{\wey}[1]{#1^w}
\newcommand{\bu}[1]{{\textbf{\textup{#1}}}}

\newtheorem{thm}{Théorème}
\newtheorem{lemma}{Lemme}

\newtheorem{proposition}{Proposition}

\title{Résonances Semiclassiques Engendrées par des Croisements de Trajectoires Classiques}

\author{Kenta Higuchi}
\address{Department of Mathematical Sciences, Ritsumeikan University,\\ 
1-1-1 Noji-Higashi, Kusatsu, 525-8577, Japan}
\thanks{The author is supported by Ritsumeikan University, KENKYU-SHOREI scholarship A} 
\email{ra0039vv@ed.ritsumei.ac.jp}



\begin{document}
\selectlanguage{french}
\begin{abstract} 
Nous consid\'erons un syst\`eme $2\times2$ d'op\'erateurs de Schr\"odinger semiclassique 1D avec petites interactions par rapport au paramètre semiclassique $h$. 
Nous étudions l'asymptotique des résonances en limite semiclassique près d'une énergie non-captive pour les deux hamiltoniens classiques correspondants. 
Nous montrons l'existence de résonances de largeur $T^{-1}h\log (1/h)$, contrairement au cas scalaire, 
sous la condition que deux trajectoires classiques se croisent et composent une trajectoire périodique de période $T$. 
\end{abstract}

\maketitle

\section{Introduction}

La théorie de résonances quantiques de l'opérateur de Schrödinger
a une longue histoire (voir \cite{DyZw}). En particulier leur répartition asymptotique en limite semiclassique a été bien étudié depuis les années 80 où Helffer et Sjöstrand \cite{HeSj1} ont montré dans un cadre analytique le non-existence de résonance dans un voisinage d'une énergie non-captive de la mécanique classique correspondante (voir aussi \cite{BCD}). 
Dans le cas captif, l'asymptotique de résonances, et surtout leur partie imaginaire (largeur), a été étudié en liaison avec la géométrie de l'ensemble capté \cite{BCD2,GeSj,HeSj1,Sj}.  
Quand le potentiel n'est pas globalement analytique, la zone sans résonance dans le cas non-captif est plus petite. 
En effet, il n'y a pas de résonance avec partie imaginaire d'ordre $h\log(1/h)$, Martinez \cite{Ma2}.

Dans le cas des opérateurs à valeurs matricielles, il faut tenir compte des interactions entre des systèmes hamiltoniens correspondants à différentes valeurs propres du symbole matriciel. Elles ne sont pas négligeables surtout quand des trajectoires classiques se croisent. Fujiié, Martinez et Watanabe  ont traité un croisement entre une trajectoire périodique et une non-captive et obtenu la largeur des résonances qui décrit l'effet du croisement (voir \cite{FMW3} et les références là-dedans).

Motivé par ces résultats, nous nous sommes posé la question: Y a-t-il des résonances engendrées par le croisement de deux trajectoires non-captives? Ce travail donne une réponse affirmative à cette question en précisant l'asymptotique des résonances d'un modèle unidimensionel provenant de l'approximation de Born-Oppenheimer.
Ce modèle a deux trajectoires classiques non-captives qui se croisent et composent une trajectoire périodique. 
Nous démontrons l'existence de résonances engendrées par cette trajectoire périodique dont la largeur est $T^{-1}h\log(1/h)$ 
contrairement au cas scalaire non-captif. Ici, la constante $T$ est la période de la trajectoire. Ce résultat optimise la zone sans résonance obtenu dans \cite{Hi}.

\section{Énoncé du résultat}
Nous consid\'erons le syst\`eme d'op\'erateurs de Schr\"odinger unidimensionnel suivant:
\begin{equation}\label{eq:SchOp}
\begin{aligned}
P(h) &=
\begin{pmatrix}
 P_1 & h W\\
h W^* & P_2
\end{pmatrix}
\end{aligned}
\quad \text{défini sur}\quad L^2(\R)\oplus L^2(\R)
\cong L^2(\R;\bbC^2),
\end{equation}
o\`u $h$ est un param\`etre semiclassique, 
$P_j=h^2D_x^2+V_j(x)$ $(j=1,2)$ 
un op\'erateur de Schr\"odinger avec un potentiel $V_j\in C^\infty(\R;\R)$, $D_x=-id/dx$. 
Ici la perturbation est donnée par un opérateur différentiel d'ordre 1, $W=W(x,hD_x)$ et son adjoint formel $W^*$. 
Nous étudions les résonances de $P(h)$ en limite semiclassique $h\to0_+$ dans un voisinage de $E_0>0$ telle que $V_1$, $V_2$, $W$ vérifient les hypothèses suivantes \textbf{(A1)}--\textbf{(A3)}:

\vspace{0.2cm}
\noindent
\textbf{(A1)}
Chaque $V_j(x)$ ($j=1,2$) vérifie les conditions suivantes:
\begin{enumerate}
\item[a.] $V_j$ se prolonge en une fonction bornée et analytique dans un domaine conique complexe
$$
\Sigma=
\bigl\{x\in\bbC\, ; \,  |\im x|<(\tan\theta_0)|\re x|,\ |\re x|> R_0\bigr\}
$$
pour des constantes $0<\theta_0<\pi/2$ et $R_0>0$. 
\item[b.] $V_j$ admet des limites $v_j^\pm$ quand $\re x\to\pm\infty$ dans $\Sigma$ et $v_1^-<E_0<v_1^+,$ $v_2^->E_0>v_2^+$.
\item[c.] $V_j-E_0$ admet un unique zéro $x_j\in\R$, et $V_j'(x_j)\neq0$. $V_j$ est analytique près de $x_j$.
\end{enumerate}

\vspace{0.2cm}
\noindent
\textbf{(A2)}
$V_1(x)=V_2(x)$ si et seulement si $x=0$.  De plus   $V_1(0)=V_2(0)=0$ et $V_1'(0)> V_2'(0).$ 

\vspace{0.2cm}
\noindent
\textbf{(A3)}
$W(x,hD_x)=r_0(x)+ir_1(x)hD_x,$ où $r_0,r_1\in C^\infty(\R;\R)$  
se prolongent en des fonctions bornées et analytiques dans $\Sigma$. De plus, $W_0:=r_0(0)+ir_1(0)\sqrt{E_0}\neq0$.


 \begin{figure}[tbp]
 \includegraphics[width=0.5\linewidth]{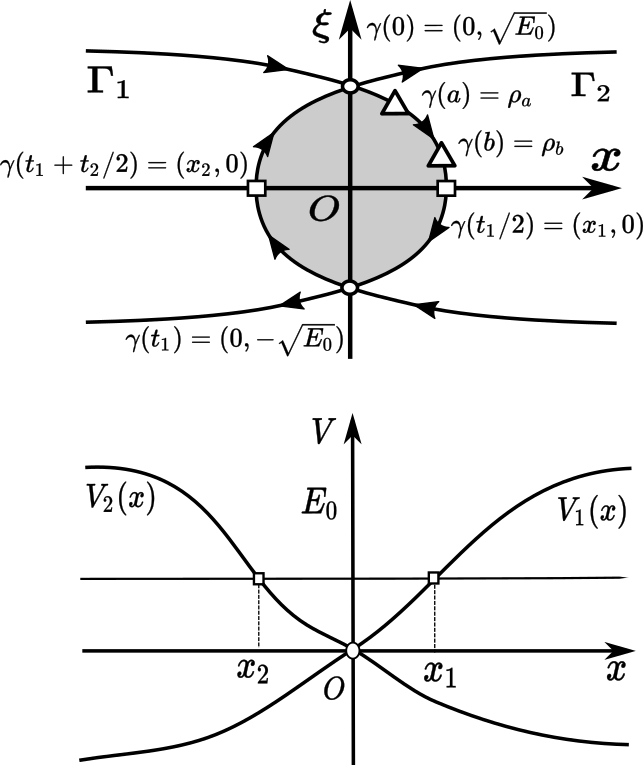}
\bigskip
\begin{center}
 \caption{Potentiels $V_1, V_2$ et trajectoires hamiltoniennes $\Gamma_1, \Gamma_2$.}
\end{center}
 \label{Fig:ex}
 \end{figure}

\vspace{0.2cm}
Sous l'hypothèse \textbf{(A1)}, l'énergie $E_0$ est non-captive pour chaque hamiltonien classique $p_j=\xi^2+V_j$ associé à $P_j$, c'est à dire que, pour chaque ensemble compact $K\subset p_j^{-1}(E_0)$, il existe $T_K>0$ telle que
\begin{equation}
\label{NT} (x,\xi)\in K\Longrightarrow  {\rm exp}(tH_{p_j})(x,\xi) \not \in p_j^{-1}(E_0)\setminus K,\ |t|>T_K.
\end{equation}
Ici, $H_{p_j}=2\xi \p_x - V'_j(x) \p_\xi$ est le champ hamiltonien, et $\exp(tH_{p_j})(x,\xi)$ est le flot hamiltonien correspondant dans l'espace des phases $\R_x\times \R_\xi$ à $p_j$ vérifiant $\exp(tH_{p_j})(x,\xi)|_{t=0}=(x,\xi)$.

Nous fixons $r>0$ arbitrairement grand et cherchons des résonances dans le rectangle 
\be\label{eq:Rec}
R_h(M):
=E_0+h\log(1/h)[-r,r]+i[-Mh\log(1/h),0],
\ee
pour une constante positive $M>0$ (indépendante de $h$). 

Dans $R_h(M)$, les résonances de $P$ sont définies, en tant que valeurs propres de l'opérateur distordu $P_\theta=U_\theta PU_\theta^{-1}$. Ici, $U_\theta u=\left|\zeta_\theta'(x)\right|^{1/2}u\circ\zeta_\theta$ où $\zeta_\theta\in C^\infty(\R;\bbC)$ vérifie 
$$
\zeta_\theta(x)=\left\{
\begin{aligned}
&x &|x|\le R_0,&\\
&xe^{i\theta} &|x|\ge2R_0,&
\end{aligned}\right.\quad0<\theta<\theta_0.
$$
Pour la méthode des distorsions analytiques, voir par exemple \cite[Section~2.7]{DyZw} et les références là-dedans. 
Remarquons que le spectre de $P_\theta$ est discret dans $R_h(M)$ dû à la condition $v_j^\pm\neq E_0$ (voir \cite[Appendix~B]{Hi}). De plus, il est indépendant de $\theta=M'h\log(1/h)$ pour $M'$ assez grand. L'ensemble de ces résonances est noté $\ope{Res}(P)$.

Soit $x_j(E)$ l'unique zéro de $E-V_j(x)$ pour $E\in R_h(M)$. 
On définit des fonctions d'énergie $E$ qui vont donner la condition de quantification des résonances: 
\begin{align}\label{eq:Const}
S(E):=2\left(\int_{x_2(E)}^0\sqrt{E-V_2(x)}dx
+\int_0^{x_1(E)}\sqrt{E-V_1(x)}dx\right),\quad
C_0(E;h):=
\frac{ihe^{iS(E)/h}\pi W_0}{(V'_1(0)-V'_2(0))\sqrt{E_0}},
\end{align}
où la branche de la racine carrée de  $E-V_j\in\bbC\setminus(-\infty,0]$ est telle que la partie réelle est positive. 
On note l'ensemble des racines de $C_0(E;h)=1$ (dites \textit{pseudo-résonances} dans \cite{BFRZbook}) par 
\be
\ope{Res}_0(P;M):=\{E\in R_h(M);\, C_0(E,h)=1\}.
\ee
Pour $E\in \ope{Res}_0(P;M)$, on a 
$$
C'_0(E;h)=ih^{-1}S'(E)C_0(E;h)=ih^{-1}S'(E)=ih^{-1}\left(S'(E_0)+\ord(h\log (1/h))\right).
$$ 
Donc pour $h$ suffisamment petit, $C'(E;h)$ ne s'annule pas. Par suite, les pseudo-résonances sont simples. 
Nous avons $E\in\ope{Res}_0(P;M)$ si et seulement si $E\in R_h(M)$ et vérifiant les deux identités
\be\label{eq:AsyPR}
\re(S(E))=h\bigl[(2n+1)\pi-\arg W_0\bigr],\quad
\im(S(E))=-h\left(\log{\frac 1h}+\log \frac{\pi|W_0|}{(V_1'(0)-V_2'(0))\sqrt{E_0}}\right),
\ee
pour un certain $n\in \N$. Par le développement en série de Taylor de $S(E)$ en $E_0$, on a 
\be\label{eq:width}
\im E=-T^{-1}h\left(\log(1/h)+\ord(1)\right)\quad\text{avec}\quad T=S'(E_0).
\ee 
Par suite, pour tout $M>T^{-1}$, l'ensemble $\ope{Res}_0(P;M)$ n'est pas vide pour $h$ suffisamment petit.

Sous les hypothèses \textbf{(A1)} et \textbf{(A2)}, il existe une ``trajectoire périodique" $\gamma$ de période $T$ (voir Figure \ref{Fig:ex}), c'est à dire, une courbe $\gamma:$ $\R\to \Gamma=p_1^{-1}(E_0)\cup p_2^{-1}(E_0)$ 
donnée par 
\be\label{eq:defgamma}
\gamma(t)=
\left\{
\begin{aligned}
&\exp(tH_{p_1})(0,\sqrt{E_0})\ &t\in[0,t_1],\\
&\exp((t-t_1)H_{p_2})(0,-\sqrt{E_0})\ &t\in[t_1,T],
\end{aligned}\right.\quad
\gamma(t+T)=\gamma(t),
\quad
t_j=\left|\int_0^{x_j}\frac{dx}{\sqrt{E_0-V_j(x)}}\right|.
\ee
Alors, $S(E_0)=\int_\gamma\xi dx$ est l'aire du domaine borné par cette trajectoire $\gamma$, et $T=t_1+t_2=\int_\gamma \xi^{-1}dx$ est sa période.

\begin{thm}\label{Theorem}
Supposons les hypothèses \bu{(A1)}--\bu{(A3)}. 
Quelle que soit $M>T^{-1}$ indépendant de $h$, il existe une résonance $E\in\ope{Res}(P)\cap R_h(M)$ de $P(h)$ pour $h$ suffisamment petit.  
Plus précisément, pour chaque $\til{E}\in\ope{Res}_0(P;M)$, il existe une résonance $E\in\ope{Res}(P)\cap R_h(M)$ telle que 
$\left|E-\til{E}\right|=o(h),$ 
et inversement, pour chaque $E\in\ope{Res}(P)\cap R_h(M)$, il existe $\til{E}\in\ope{Res}_0(P;M)$ telle que $\left|E-\til{E}\right|=o(h).$
\end{thm}

Ce théorème peut se généraliser aux cas traités dans \cite{Hi} avec plusieurs croisements.

\section{Preuve du Théorème \ref{Theorem}}
Ici, nous allons rappeler des notions de base de l'analyse microlocale et semiclassique renvoyant le lecteur aux livres \cite{DiSj, Ma3, Zw} pour plus de détails. 
L'opérateur $h$-pseudodifférentiel correspondant à un symbole $a\in \mc{S}(T^*\R)$, noté par $\wey{a}(x,hD)$, est défini sur $\mathcal{S}(\R)$ par 
\begin{align}\label{WeylQ}
\wey{a}(x,hD)u(x)
:=\frac{1}{2\pi h}\int_{\R^2}e^{i(x-y)\xi/h}a\left(\frac{x+y}{2},\xi\right)u(y)dyd\xi.
\end{align}
Soit $f=f(x;h)\in L^2(\R)$ avec $\|f\|_{L^2}\le 1$. 
On écrit $f\sim0$ microlocalement en $(x,\xi)\in T^*\R$ 
s'il existe un symbole $\chi\in \mc{S}(T^*\R)$ tel que $\chi(x,\xi)\neq0$ et $\|\wey{\chi}(x,hD)f\|_{L^2}=\ord(h^\infty).$  
On écrit $f\sim g$ si $f-g\sim0$. 

Revenons \`a notre op\'erateur $P(h)$.
Pour chaque $j=1,2$, la trajectoire classique $p_j^{-1}(E_0)$   passe par un point tournant $(x_j,0)$ et elles se croisent aux points $(0,\pm\sqrt{E_0})$. 
L'ensemble caract\'eristique $\Gamma$ priv\'e de ces 4 points
se compose de 8 courbes connexes (voir Figure \ref{Fig:ex}), et le long de chaque courbe, l'espace des solutions microlocales de $(P(h)-E)u\sim 0$ pour $E\in R_h(M)$ est de dimension 1 engendr\'e par une solution BKW.
Sur la partie $\gamma(]0,t_1/2[)$, par exemple, elle est de la forme
\be\label{eq:WKB}
\begin{aligned}
&w\sim 
e^{i\phi(x)/h}
\left (
\begin{array}{c}
\sigma_1 \\
h\sigma_2
\end{array}
\right ),
\quad\text{o\`u}\quad
\phi(x):= \int_0^x \sqrt{E-V_1(t)} \,dt,
\quad
\sigma_j=(1+\ord(h))\sigma_{j,0}(x,E)=\ord(1),\\
&\text{avec}\quad
\sigma_{1,0}(x,E)=\frac{1}{(E-V_1(x))^{1/4}},\quad
\sigma_{2,0}(x,E)=\frac{r_0(x)-ir_1(x)\sqrt{E-V_1(x)}}{(V_1(x)-V_2(x))(E-V_1(x))^{1/4}}.
\end{aligned}
\ee
Ici, $w$ est analytique pour $E\in R_h(M)$. De plus, on a $w=\ord(h^{-N_1})$ pour un certain $N_1>0$ uniformément sur chaque compact dans $]0,x_1[$ car $\left|\im E\right|\le Mh\log(1/h)$.

On prend deux points $\rho_a=\gamma(a)$, $\rho_b=\gamma(b)$, $0<a<b<t_1/2$ sur la courbe $\gamma(]0,t_1/2[)$ (voir Figure 1).
On prolonge une solution microlocale $u$  en $\rho_b$ jusqu'en $\rho_a$ dans les deux directions le long de $\gamma$.
D'abord, on la prolonge dans la direction négative le long de $\gamma([a,b])$. Si $(P-E)u\sim 0$ microlocalement pr\`es de cette courbe et si $u\sim \alpha_b w$ en $\rho_b$, $u\sim \alpha_a w$ en $\rho_a$, alors \'evidemment  $\alpha_a=\alpha_b$. Maintenant, on prolonge la solution dans la direction positive le long de $\gamma([b, T+a])$ en supposant que $(P-E)u\sim 0$ pr\`es de cette courbe, notant que l'on passe par deux points de croisement et deux points tournants.
D'après \cite[Sec.~5]{FMW3}, on a le lemme suivant qui dit que la solution microlocale en $\rho_a$ est uniquement d\'etermin\'ee par celle en $\rho_b$ et donne une relation explicite au niveau principal, sous la condition que $u\sim 0$ sur les courbes entrantes $\Gamma_{\ope{ent}}^1=p_1^{-1}(E_0)\cap\{x<0,\xi>0\}$ et $\Gamma_{\ope{ent}}^2=p_2^{-1}(E_0)\cap\{x>0,\xi<0\}$, la solution microlocale en $\rho_a$ est uniquement déterminée (voir aussi \cite[Sec.~3.2.1]{Hi}). En utilisant les trois arguments qui suivent : 
\begin{itemize}
\item le théorème de propagation de singularit\'es standard \cite[Theorem~12.5]{Zw}, valide aussi dans le cas des systèmes, 
\item la th\'eorie de Maslov le long des trajectoires hamiltoniennes du type principale \cite[Lemma~6.1]{FMW3}, qui permet de définir la solution près des points tournants, 
\item l'\'etude microlocale des solutions au voisinage des points de croisements \cite[Propositions~5.7 and 5.9]{FMW3}, qui permet d'étendre la solution à travers ses points,
\end{itemize}
nous obtenons le résultat suivant :
\begin{lemma}[{\cite[Section~3.2.1]{Hi}}]\label{lem:Mono}
Supposons que $u\in L^2(\R;\bbC^2)$ avec $\|u\|_{L^2}\le1$ vérifie $(P(h)-E)u\sim 0$ microlocalement dans un voisinage de $\gamma([b,T+a])$ et $u\sim 0$ le long de $\Gamma_{\ope{ent}}^1\cup\Gamma_{\ope{ent}}^2$. 
Si $u\sim\alpha_b w$ microlocalement en $\rho_b$, alors $u\sim\alpha_a w\text{ microlocalement en }\rho_a$, 
où $\alpha_a$ est donné par 
\be\label{eq:monodromie}
\alpha_a=C(E,h)\alpha_b\quad\text{avec}\quad C(E,h)=C_0(E,h)(1+{\mathcal O}(h\log(1/h))).
\ee
\end{lemma}

La formule \eqref{eq:Const} et l'estimation \eqref{eq:monodromie} de l'erreur impliquent que les racines de $C(E;h)=1$ sont 
à l'ordre $h^2\log(1/h)$ approchées par celles de $C_0(E;h)=1$.  
En effet, nous avons l'estimation suivante :
\begin{equation*}
|C_0(E;h)-1|=|(E-E_1)C_0'(E_1;h)+\ord(h^{-2}(E-E_1)^2)|\gtrsim  h^{-1}|E-E_1|,
\end{equation*}
pour $E_1\in \ope{Res}_0(P;M)$ et $E\in R_h(M)$ avec $|E-E_1|=o(h)$, et le théorème de Rouché nous permet de conclure.


Pour la démonstration de Théorème \ref{Theorem} on applique la méthode microlocale établie dans \cite{BFRZbook} pour les résonances créées par des trajectoires homoclines et hétéroclines d'un opérateur de Schrödinger scalaire multidimensionnel. 
Les arguments de \cite{BFRZbook} restent valable tant que $\theta=\ord(h\log(1/h))$, et donc s'appliquent à notre problème. 

On écrit $B(z;r)=\{\zeta\in\bbC;\,\left|\zeta-z\right|<r\}$ $(z\in\bbC,\,r>0)$.  
\begin{proposition}\label{prop1}
Quelles que soient $M,\e>0$, il n'existe pas de résonance de $P$ dans le domaine 
\be\label{eq:FreeDom}
R_h^\e(M):=R_h(M)\setminus (\ope{Res}_0(P;M)+B(0,\e h)),
\ee
pour $h>0$ assez petit. 
Plus précisément, il existe un nombre $N>0$ tel que
\be\label{eq:ResEsti}
\|(P_\theta-E)^{-1}\|= \ord(h^{-N}),
\ee
uniformément pour $E\in R_h^\e(M)$. 
Ici, $\left\|\cdot\right\|$ est la norme des opérateurs bornés définis sur $L^2(\R;\bbC^2)$.
\end{proposition}

\begin{proof}
On démontre cette proposition par l'absurde. Supposons que l'estimation \eqref{eq:ResEsti} ne soit pas vraie. Alors, il existerait 
une suite de $h$, $E$ et $u$ telle que $h$ tend vers $0$, $E\in R_h^\e$ et $u\in H^2(\R;\bbC^2)$ avec $\|u\|_{L^2(\R;\bbC^2)}=1$ et 
\be\label{eq:QuasiM}
\|(P_\theta-E)u\|_{L^2(\R;\bbC^2)}=\ord(h^\infty).
\ee 
La condition \eqref{eq:QuasiM} implique que $u\sim0$ sur les courbes entrantes $\Gamma_{\ope{ent}}^1\cup\Gamma_{\ope{ent}}^2$. 
Donc si on montre que $u\sim0$ microlocalement près de l'ensemble capté $\gamma([0,T])$, alors on en déduirait $\|u\|_{L^2(\R;\bbC^2)}=\ord(h^\infty)$ par le même argument que \cite[Sec.~8]{BFRZbook}, d'où la contradiction. 

Soient $\alpha_a,\alpha_b$ tels que $u\sim\alpha_{a} w$ microlocalement en $\rho_{a}$ et $u\sim\alpha_{b} w$ microlocalement en $\rho_{b}$. 
Puisque $(P-E)u\sim0$ le long de $\gamma$, Lemme~\ref{lem:Mono} implique
$(1-C(E;h))\alpha_b=\ord(h^\infty).$ 
Pour $\e>0$ fixé, il existe $c(\e)>0$ indépendant de $h$ tel que 
$\left|1-C(E;h)\right|\ge c(\e)$ uniformément pour $E\in R_h^\e(M)$. 
Cela implique que $\alpha_b=\ord(h^\infty)$, et par conséquence, $u\sim0$ microlocalement en $\rho_b$. 
Encore une fois par le théorème de propagation des singularités le long de $\gamma$, on en déduit que $u\sim0$ microlocalement près de l'ensemble capté.
\end{proof}

\begin{proposition}\label{prop2}
Pour tout $M>T^{-1}$, $\e>0$ et $\til{E}\in\ope{Res}_0(P;M)$, il existe au moins une résonance dans $B(\til{E},\e h)$ pour $h$ suffisamment petit.
\end{proposition}
\begin{proof}
Soit $\til{E}\in \ope{Res}_0(P;M)$. 
Prenons $\dl>0$ petit et des troncatures $\chi,\psi\in C_0^\infty(T^*\R)$ supportées dans un voisinage de $\gamma(]0,t_1/2[)$ vérifiant 
\begin{align*}
&
\chi=1\quad\text{sur }\ \gamma([a-2\dl,a+\dl])
,\quad
\ope{supp}\chi\cap \gamma([0,t_1/2])\subset \gamma(]0,b-\dl]),\\
&
\psi=1\quad\text{sur }\ 
\ope{supp}(\nabla\chi)\cap \gamma(]0,a-2\dl]),
\quad
\ope{supp}\psi\cap\gamma([0,t_1/2])\subset \gamma(]0,a-\dl]).
\end{align*}
On définit la fonction test $v$ par 
$v:=\psi^w[P,\chi^w]w,$ et $u(x,h;E):=(P_\theta-E)^{-1}v$ pour $E\in\p B(\til{E},\e h)$. 
D'après \eqref{eq:WKB}, $w$ est analytique et $\ord(h^{-N_1})$ pour $E\in B(\til{E},\e h)$, donc $v$ l'est aussi. 
L'estimation \eqref{eq:ResEsti} montre que $u$ est bien définie et $\|u\|_{L^2}\lesssim h^{-(N+N_1)}$ uniformément sur $\p B(\til{E},\e h)$. 
Posons $u_{\ope{test}}:=\chi^w w$. Alors il vérifie $(P-E)u_{\ope{test}}\sim v$ le long de $\gamma([b,T+a])$, $u_{\ope{test}}\sim 0$ en $\rho_b$ et $u_{\ope{test}}\sim w$ en $\rho_a$. 
En effet, on a $[P,\chi^w]\sim0$ en dehors du $\ope{supp}(\nabla\chi)$ et 
\begin{align*}
(P-E)u_{\ope{test}}
=
(P-E)\chi^ww
\sim \psi^w[P,\chi^w]w
=v\ 
\text{microlocalement près de }\gamma([b,T+a]). 
\end{align*}
Par contre, $u_{\ope{hom}}:=u-u_{\ope{test}}$ vérifie $(P-E)u_{\ope{hom}}\sim0$ le long de $\gamma([b,T+a])$. 
D'après Lemme~\ref{lem:Mono}, on a $\alpha_a=C\alpha_b$, où $\alpha_a,$ $\alpha_b$ sont les constantes telles que $u_{\ope{hom}}\sim\alpha_b w$ en $\rho_b$ et $u_{\ope{hom}}\sim \alpha_aw$ en $\rho_a$. 
Comme $v\sim0$ microlocalement près de $\gamma([a,b])$, $u$ est une solution de l'équation homogène. Donc on a $u=u_{\ope{hom}}+u_{\ope{test}}\sim C\alpha_bw+w\sim Cu+w$ microlocalement en $\rho_b$, c'est à dire,
\be\label{eq:TestShift}
u\sim (1-C)^{-1} w
\quad\text{microlocalement en }\rho_b.
\ee
Comme il existe une unique racine simple $\til{E}_0$ de $C(E;h)=1$ dans $B(\til{E},\e h)$, la valeur
\be\label{eq:Oint1}
\oint_{\p B(\til{E},\e h)}\,\,\frac{w(x,h;E)}{1-C(E;h)}dE
=-2\pi i \frac{w(x,h;\til{E}_0)}{C'(\til{E}_0;h)}
\ee
ne s'annule pas. Ici, on utilise que $C'(E_0;h)$ est non nulle pour $h$ suffisamment petit. 
Par contre, s'il n'existait pas de résonance, un pôle de $(P_\theta-E)^{-1}$, dans $B(\til{E},\e h)$, on aurait aussi
\be\label{eq:Oint2}
\oint_{\p B(\til{E},\e h)}\,\,u(x,h;E)dE
=\oint_{\p B(\til{E},\e h)}\,\,(P_\theta-E)^{-1}v(x,h;E)dE=0.
\ee
Les égalités \eqref{eq:Oint1} et \eqref{eq:Oint2} sont contradictoires avec l'égalité \eqref{eq:TestShift}. 
Par conséquence, il doit y avoir au moins une résonance dans la boule $B(\til{E},\e h)$.
\end{proof}


\bibliography{ResGenbib1}
\bibliographystyle{unsrt}
\nocite{BCD2}

\end{document}